\documentclass[runningheads]{llncs}
\usepackage[T1]{fontenc}
\usepackage{longtable}
\usepackage{booktabs}

\usepackage{longtable}
\usepackage{hyperref}
\usepackage{amssymb,amsmath,amsfonts}
\usepackage{mathrsfs}
\usepackage{color}

\usepackage[table]{xcolor}
\usepackage{makecell}
\usepackage{tabularx}
\usepackage{array}
\usepackage{multirow}
\usepackage{colortbl}

\definecolor{lightgreen}{RGB}{226,239,218}
\definecolor{lightred}{RGB}{244,204,204}
\definecolor{lightyellow}{RGB}{255,242,204}
\definecolor{headerblue}{RGB}{220,230,242}

\definecolor{headerblue}{RGB}{151,78,161}
\definecolor{softblue}{RGB}{232,240,249}
\definecolor{softgreen}{RGB}{172,245,193}
\definecolor{softred}{RGB}{240,188,218}
\definecolor{softyellow}{RGB}{255,249,219}
\definecolor{rowgray}{RGB}{208,248,248}
\definecolor{softred2}{RGB}{150,198,200}

\newcommand{\yescell}{\cellcolor{softgreen}\textsf{FPT}}
\newcommand{\nocell}{\cellcolor{softred}\textsf{NP}-hard}
\newcommand{\polcell}{\cellcolor{softgreen}\textsf{Poly-time}}

\newcommand{\paracell}{\cellcolor{softred}\textsf{para-NP}-hard}
\newcommand{\egalcell}{\cellcolor{softred2}{\texttt{EGAL}: {\sf para-NP}-hard \texttt{UTIL}: {\sf W[2]}-hard, {\sf XP}}}
\newcommand{\utiltd}{\cellcolor{softyellow}{\texttt{EGAL}: {\sf para-NP}-hard \texttt{UTIL}: \FPT}}
\newcommand{\egaltw}{\cellcolor{softyellow}{\texttt{EGAL}: {\sf para-NP}-hard \texttt{UTIL}: \FPT}}

\usepackage{mathtools}
\usepackage{bbm}
\usepackage{setspace}

\usepackage{tikz}
\usetikzlibrary{arrows,positioning}

\newdimen\prevdp
\def\leftlabel#1{\noalign{\prevdp=\prevdepth
   \kern-\prevdp\nointerlineskip\vbox to0pt{\vss\hbox{\ensuremath{#1}}}\kern\prevdp}}

\usepackage{url}

\usepackage{enumerate}

\usepackage{xspace}

\newcommand{\NPC}{\ensuremath{\mathsf{NP}}\text{-complete}\xspace}
\newcommand{\NPH}{\ensuremath{\mathsf{NP}}\text{-hard}\xspace}

\newcommand{\WOH}{\ensuremath{\mathsf{W[1]}}-hard\xspace}

\newcommand{\WTH}{\ensuremath{\mathsf{W[2]}}-hard\xspace}
\newcommand{\WTC}{\ensuremath{\mathsf{W[2]}}-complete\xspace}
\newcommand{\FPT}{\ensuremath{\mathsf{FPT}}\xspace}
\newcommand{\XP}{\ensuremath{\mathsf{XP}}\xspace}

\newcommand{\tw}{\ensuremath{\mathrm{tw}}\xspace}

\let\oldlambda\lambda
\renewcommand{\lambda}{\ensuremath{\oldlambda}\xspace}
\let\oldalpha\alpha
\renewcommand{\alpha}{\ensuremath{\oldalpha}\xspace}
\let\oldDelta\Delta
\renewcommand{\Delta}{\ensuremath{\oldDelta}\xspace}

\newcommand{\OO}{\ensuremath{\mathcal O}\xspace}

\newcommand{\NB}{\ensuremath{\mathbb N}\xspace}

\usepackage{nicefrac}

\newcommand{\ignore}[1]{}

\renewcommand{\leq}{\leqslant}
\renewcommand{\geq}{\geqslant}
\renewcommand{\ge}{\geqslant}
\renewcommand{\le}{\leqslant}

\usepackage{enumitem}
\setlist[enumerate]{labelwidth=!, labelindent=0pt}

\usepackage{rotating}
\usepackage{amssymb}
\usepackage{latexsym}
\usepackage{epsfig}
\usepackage{xcolor}
\usepackage{url}
\usepackage{caption}
\usepackage{subcaption}
\usepackage{array}
\usepackage{multirow}
\usepackage{enumerate}
\usepackage{pdflscape}
\usepackage{amsmath,graphicx,algorithm,algpseudocode,amsfonts}
\usepackage{epstopdf}
\usepackage{float}
\floatstyle{ruled}
\newfloat{Algorithms}{!thb}{lop}
\floatname{Algorithms}{Algorithm}

\allowdisplaybreaks[1]

\algnewcommand\algorithmicinput{\textbf{Input:}}
\algnewcommand\INPUT{\item[\algorithmicinput]}
\algnewcommand\algorithmicoutput{\textbf{Output:}}
\algnewcommand\OUTPUT{\item[\algorithmicoutput]}
\algnewcommand{\LineComment}[1]{\State \(\triangleright\) #1}

\usepackage{pifont}

\usepackage{mathtools}

\usepackage{algorithm}
\usepackage{algpseudocode}


\usepackage{comment}

\usepackage{cleveref}
\definecolor{cblue}{RGB}{0, 0, 128}
\crefname{theorem}{Theorem}{\bf Theorems}
\crefname{observation}{Observation}{\bf Observations}
\crefname{corollary}{Corollary}{\bf Corollary}
\crefname{lemma}{Lemma}{\bf Lemmata}
\crefname{corollary}{Corollary}{\bf Corollaries}
\crefname{proposition}{Proposition}{\bf Propositions}
\crefname{definition}{Definition}{\bf Definitions}
\crefname{claim}{Claim}{\bf Claims}
\crefname{reductionrule}{Reduction rule}{\bf Reduction rules}
\usepackage{color,soul} 

\usepackage{todonotes}
\newcommand{\CEGW}{\textsc{Connected Egalitarian Graph Welfare}}
\newcommand{\CUGW}{\textsc{Connected Utilitarian Graph Welfare}}

\begin{document}
\title{Fair, Efficient and Connected Allocations on Graphs}
\titlerunning{Fair, Efficient and Connected Allocations on
Graphs}
%

%

\author{Susobhan Bandopadhyay\inst{1}\and
Anish Datta \inst{2} \and
Palash Dey \inst{2} \and
Ashlesha Hota \inst{2} \and
Abhishek Sahu \inst{3}
}

 \authorrunning{S. Bandopadhyay et al.}
 \institute{Tata Institute of Fundamental Research Mumbai, India \and
 Indian Institute of Technology Kharagpur, India \and
National Institute of Science and Educational Research Bhubaneshwar, India}
\maketitle              
\begin{abstract}

We study the classical and parameterized complexity of efficient connected allocation problems on graphs, where efficiency is measured by egalitarian and utilitarian welfare maximization. We first establish a sharp complexity dichotomy in the classical setting: both problems are \NPH in general and remain hard even on very restricted graph classes such as paths, and consequently trees and cycles. In contrast, they are polynomial-time solvable on stars, but this tractability does not extend even to the case of two disjoint stars. 
Motivated by these boundaries, we move to the parameterized complexity framework, where we study the problem with respect to the number of agents. We obtain fixed-parameter tractability (\FPT) on trees and, more generally, identify a robust phenomenon whereby tractability on a connected graph class extends to disjoint unions of graphs from that class. We further investigate the parameters treewidth and treedepth, showing that the utilitarian version is \FPT for both, whereas the egalitarian version remains {\sf para\text{-}NP}-hard even on graphs of treedepth two.
Finally, we analyze the number of connected components and show that except for the collection of stars, the problems remain hard. For the collection of stars, while we obtain {\sf para-NP}-hardness for the egalitarian case, the utilitarian case gives  {\sf W[2]}-hardness together with an {\sf XP} algorithm. 



\end{abstract}

\section{Introduction}
Fair division of indivisible goods is a fundamental topic in economics, social choice, and computer science~\cite{DBLP:journals/ai/AmanatidisABFLMVW23}. It studies how to allocate indivisible goods among agents with different valuations, under various models such as binary, additive, and submodular valuations. Several notions of fairness and allocation objectives have been extensively studied.
A particularly interesting variant arises when goods are connected through an underlying structure such as a graph, matroid, or network. In many real-world settings, goods have natural dependencies or connectivity constraints. One important framework is the \emph{fair division of graphs}, where goods correspond to graph vertices and allocations must be connected. For example, when dividing land among individuals, a person would usually prefer one connected plot rather than several disconnected pieces. Similar constraints arise in task allocation and networked goods, making connectivity a natural requirement in fair allocation problems.

In the \emph{connected fair division} model, each good is represented as a vertex of a graph, and each agent must receive a connected subgraph as her bundle. This connectivity requirement makes the problem significantly harder. While fair or efficient allocations can often be found in polynomial time in the unconstrained setting, adding graph-based feasibility constraints typically makes the problem \NPH, even for simple graph classes such as paths~\cite{kawase2024contiguous}. Several prior works have studied fairness notions under connectivity constraints. For example,~\cite{bouveret2017fair} characterized the complexity of proportionality and envy-freeness on paths, trees, and stars; 
~\cite{lonc2019maximinshareallocationscycles} studied maximin share allocations on cycles; 
~\cite{igarashi2019pareto} examined Pareto-optimal allocations under connectivity constraints; 
and~\cite{BILO2022197} analyzed almost envy-free allocations on paths and related graph classes.
In contrast, the problem of computing allocations that optimize social welfare objectives such as utilitarian and egalitarian welfare remains largely unexplored, and we study this direction.

We study two classical welfare objectives in this paper: \emph{egalitarian welfare}, which seeks to maximize the minimum utility among all agents, and \emph{utilitarian welfare}, which seeks to maximize the total utility across all agents. We investigate the computational complexity of finding connected allocations that optimize these objectives. For an allocation $\mathcal{A}$, the \emph{egalitarian welfare} is defined as $\mathrm{EW}(\mathcal{A})=\min_{i\in N}u_i(A_i)$, while the \emph{utilitarian welfare} is defined as $\mathrm{UW}(\mathcal{A})=\sum_{i\in N}u_i(A_i)$. Formally, we study the following two problems.

\begin{definition}[\CUGW]
Given an undirected graph $G=(V,E)$, a set of agents $N$, additive valuation functions $u_i : V \to \mathbb{Z}_{\ge 0}$ for each $i\in N$, and a threshold $t\in\mathbb{Z}_{\ge 0}$, determine whether there exists an allocation $\mathcal{A}=(A_1,\ldots,A_n)$ of $V$ such that $(A_1,\ldots,A_n)$ forms a partition of $V$, each $G[A_i]$ is connected, and the utilitarian welfare satisfies
\(UW(\mathcal{A})=\sum_{i\in N}u_i(A_i)\ge t\).
\end{definition}

\begin{definition}[\CEGW]
Given an undirected graph $G=(V,E)$, a set of agents $N$, additive valuation functions $u_i : V \to \mathbb{Z}_{\ge 0}$ for each $i\in N$, and a threshold $t\in\mathbb{Z}_{\ge 0}$, determine whether there exists an allocation $\mathcal{A}=(A_1,\ldots,A_n)$ of $V$ such that $(A_1,\ldots,A_n)$ forms a partition of $V$, each $G[A_i]$ is connected, and the egalitarian welfare satisfies
\(EW(\mathcal{A})=\min_{i\in N}u_i(A_i)\ge t\).
\end{definition}

The interaction between these welfare notions and graph-based feasibility constraints gives rise to rich algorithmic and structural questions. 
Understanding the boundary between tractable and intractable instances across various graph classes (such as trees, paths, stars, and cliques) is crucial for a comprehensive characterization of the computational complexity of fair division with connectivity constraints.

In this work, we study the complexity of maximizing egalitarian and utilitarian welfare under connected allocations. 
We provide tight hardness and tractability results across fundamental graph classes and parameterizations, contributing to a unified understanding of the algorithmic frontier of connected fair division.

\begin{table}[t]
\centering
\small
\caption{Computational complexity of unparameterized and parameterized variants for various graph structures.}
\label{tab:complexity}

\setlength{\tabcolsep}{6pt}
\renewcommand{\arraystretch}{1.25}
\rowcolors{2}{white}{rowgray}

\begin{tabularx}{\linewidth}{
>{\raggedright\arraybackslash}X
>{\centering\arraybackslash}p{1.7cm}
>{\centering\arraybackslash}p{2.9cm}
>{\centering\arraybackslash}p{2.95cm}
}
\toprule
\rowcolor{headerblue}
\textcolor{white}{\textbf{Structure}} &
\textcolor{white}{\textbf{Classical Complexity}} &
\textcolor{white}{\textbf{Param. by \#Agents ($n$)}} &
\textcolor{white}{\textbf{Param. by \#Comp.($k$)}} \\
\midrule

Path, Tree, Cycle &
\nocell &
\yescell &
--- \\

Star &
\polcell &
\yescell &
--- \\

Collection of Paths/Trees/Cycles &
\nocell &
\yescell &
\paracell \\

Collection of Stars &
\nocell &
\yescell &
\egalcell \\

Bounded Treewidth &
--- &
\egaltw &
--- \\

Bounded Treedepth &
--- &
\utiltd &
--- \\

\bottomrule
\end{tabularx}
\end{table}

\subsection{Our contributions}

We study efficient connected allocation problems on graphs, where the following two well-studied welfare measures serve as the notion of efficiency: {\em egalitarian} welfare and {\em utilitarian} welfare. We investigate both the classical and parameterized complexity of these problems for several graph classes and natural parameters. Our main results are shown in \Cref{tab:complexity}. We now provide a high level overview of our results. 


\noindent\textbf{Classical Complexity.} These problems are \NPH in general. We establish that this hardness persists even on the simplest graph classes: paths, and consequently on trees and cycles. In contrast, we show that they are polynomial-time solvable on stars; however, this tractability does not extend to disjoint collections of stars. Motivated by these boundaries and to address the intractability, we study the parameterized complexity of the problems with respect to two natural parameters: the number of agents and the number of connected components. Notably, the latter is precisely the parameter that determines tractability for disjoint unions of stars. Across both objectives, we obtain a striking picture: while a small number of agents often leads to tractability, a small number of components alone does not suffice for fixed-parameter tractability.


\noindent\textbf{Parameterized by number of agents.} We prove that both \CEGW{} and \CUGW{} problems are fixed-parameter tractable on trees and forests when parameterized by the number of agents. While dynamic programming techniques suffice on trees, extending these algorithms to forests presents an additional challenge since connected allocations cannot span multiple components. We address this issue by observing that every feasible allocation induces a partition of the agent set according to the connected components utilized by the allocation. Our algorithm combines enumeration of these agent partitions with color-coding techniques that efficiently identify suitable component assignments. This avoids exhaustive exploration of all assignments and yields randomized \FPT algorithms, which can further be derandomized using standard perfect hash family constructions. 
 
Notably, during the design of our algorithm we observe that the techniques employed do not rely on any specific properties required for a tree. This implies a more general phenomenon: whenever a problem is \FPT when parameterized by the number of agents on a connected graph class $\mathcal{G}$, it remains \FPT on disjoint collections of graphs from $\mathcal{G}$. In contrast to the classical complexity behavior observed on stars, where tractability fails to extend to a disjoint collection, the parameterized setting   (number of agents) imposes a robust property.

We further investigate how tree-like structures affect fixed-parameter tractability. Building on our results for trees and forests, we explore two of the most natural parameters measuring similarity with a tree structure: treewidth and treedepth, obtaining contrasting results: the \CUGW{} problem is \FPT for graphs with bounded treewidth and treedepth, whereas the \CEGW{} problem remains {\sf para-NP}-hard on both graph classes. The \FPT result for the \CUGW{} problem follows from a tree-decomposition-based dynamic programming (bounded treewidth graphs); since treedepth is an upper bound on treewidth, the result extends to both parameters. In contrast, a reduction from the \textsc{Partition} problem establishes that the \CEGW{} problem is \NPH even for graphs with treedepth two, revealing a tractability gap between the two versions.

\noindent\textbf{Parameterized by number of components.}  We also study the influence of the number of connected components as a structural parameter. On all the discussed simple graph classes (excluding stars), both problems are \NPH, which immediately dismisses tractability when parameterized by the number of components. Only on star collections do we observe a more nuanced picture. For the egalitarian setting, we establish \NPH{ness} on graphs consisting of two disjoint stars through a reduction from the \textsc{Partition} problem establishing the stronger {\sf para-NP}-hardness result. For utilitarian welfare, we prove a much weaker \WTH{ness} on collections of stars via a reduction from the \textsc{Set Cover} problem. And we are able to design an \XP algorithm for the \CUGW{} problem parameterized by the number of stars ruling out the possibility of the stronger {\sf para-NP}-hardness for the utilitarian version. The algorithm for the \CUGW{} problem guesses ownership patterns of star centers and transforms the remaining optimization task into a maximum-weight bipartite matching problem.

\subsection{Related Work}
The {\sc Equitable Connected Partition} problem studied in~\cite{enciso2009makes} provides structural insights into how connectivity constraints influence the complexity of partitioning a graph into balanced connected components. In this problem, the vertices are divided into a fixed number of connected classes whose sizes differ by at most one. From a parameterized complexity viewpoint, the problem is \WOH when parameterized by the number of classes combined with treewidth, pathwidth, or feedback vertex set size (even on planar graphs), but is \FPT when parameterized by vertex cover size or by the maximum number of leaves in a spanning tree. These results highlight how specific structural parameters of the graph can critically affect tractability under connectivity requirements.

Bouveret et al.~\cite{bouveret2017fair} introduced and studied the model of fair allocation of indivisible items under connectivity constraints, where items form an undirected graph and each agent must receive a connected subgraph. Focusing on additive valuations, they analyzed classical fairness notions such as proportionality, envy-freeness, and maximin share (MMS). They established strong hardness results even for simple graph classes (e.g., \NPH{ness} of proportionality on paths and envy-freeness on paths and stars), while also providing efficient algorithms for restricted settings, such as graphs with simple structure or instances with a small number of agents or agent types. In particular, they showed that MMS allocations always exist and can be computed efficiently for acyclic graphs. Igarashi et al.~\cite{igarashi2019pareto} studied the problem of computing Pareto-optimal allocations under connectivity constraints in the item graph. They show that the problem becomes \NPH for many topologies, including trees of bounded pathwidth or maximum degree~3. They further establish that even on paths, deciding the existence of a Pareto-optimal allocation satisfying EF1 or MMS is \NPH. Nevertheless, they provide a positive result by designing a moving-knife algorithm for certain binary valuation instances with non-nested interval structure. 

Aziz et al.~\cite{aziz2022fair} studied fair division with both goods and chores, where agents may have positive or negative utilities for items, thereby generalizing the standard models that consider only goods or only chores. They investigate which classical fairness and efficiency guarantees extend to this mixed setting and identify several limitations. In addition, they design new polynomial-time algorithms for computing fair allocations and analyze the complexity of achieving various fairness and efficiency notions, highlighting gaps in the existing literature.

Further, Bil\`o et al.~\cite{BILO2022197} studied the existence of EF1 allocations under connectivity constraints in the item graph. For path graphs with monotonic utilities, they prove the existence of EF1 allocations for up to four agents and EF2 allocations for any number of agents, using discrete analogues of classical topological arguments. For identical utilities, they provide a polynomial-time algorithm to compute EF1 allocations. For two agents, they fully characterize the graphs that guarantee EF1 existence, namely those whose biconnected components form a path. Deligkas et al.~\cite{deligkas2021parameterized} provided a detailed parameterized complexity analysis under fairness notions such as proportionality, envy-freeness, EF1, and EFX. They show that tractability requires meaningful restrictions on both the structure of the item graph and the agents. In particular, they design \XP algorithms parameterized by (i) clique-width of the graph plus the number of agents and (ii) treewidth plus the number of agent types, complemented by matching lower bounds. They further show that achieving \FPT additionally requires bounding the maximum item valuation alongside stronger structural parameters.

While fair allocation under connectivity constraints has received considerable attention with respect to notions such as proportionality, envy-freeness, EF1, and EFX, comparatively little attention has been devoted to welfare-oriented objectives such as egalitarian and utilitarian optimization. In this paper, we study these objectives and investigate their computational aspects.

\section{Preliminaries}

\textbf{Notations.} We denote the set $\{1,2,\ldots\}$ of natural numbers by $\NB$. For any integer $\ell$, we denote the set $\{1,\ldots,\ell\}$ by $[\ell]$. All graphs considered in this paper are undirected and simple. For a graph $G$, we denote its vertex set and edge set by $V(G)$ and $E(G)$, respectively. For a vertex $v\in V(G)$, let $N_G(v)=\{u\mid \{u,v\}\in E(G)\}$ denote the open neighbourhood of $v$ in $G$. The degree of a vertex $v$ is denoted by $\deg_G(v)=|N_G(v)|$. A graph $G$ is called \emph{$k$-regular} if $\deg_G(v)=k$ for every vertex $v\in V(G)$. For a set $S\subseteq V(G)$, the graph induced by $S$ is denoted by $G[S]$. For a graph $G$ and a vertex $v$, let $G-v$ denote the graph $G[V(G)\setminus\{v\}]$. For a disconnected graph $G$, let $c(G)$ denote the set of connected components of $G$.

We consider a fair division setting on an undirected graph $G=(V,E)$, where the vertex set $V(G)$ represents a collection of indivisible items. Let $N=\{1,\dots,n\}$ denote a set of $n$ agents. Each agent $i\in N$ is associated with a non-negative valuation function $u_i:V(G)\rightarrow \mathbb{Z}_{\ge 0}$. Further, for any bundle $A\subseteq V$, the utility of agent $i$ for $A$ is given by $u_i(A)=\sum_{v\in A}u_i(v)$.

An allocation is an $n$-partition $\mathcal{A}=(A_1,\dots,A_n)$ of $V$, where bundle $A_i$ is assigned to agent $i$. We consider feasibility constraints on allocations: each bundle assigned to an agent must belong to a family $\mathcal{F}(G)\subseteq 2^V$. Examples include connected subsets and independent subsets of $G$. When $\mathcal{F}(G)$ is the family of connected subsets, we refer to $\mathcal{A}$ as a \emph{connected allocation}; similarly, if $\mathcal{F}(G)$ is the family of independent sets, we call it an \emph{independent allocation}.

\noindent
\textbf{Parameterized Complexity.} In parameterized complexity, an input instance of size $n$ is associated with a parameter $k$, and the goal is to design algorithms with running time $f(k)\cdot n^{\OO(1)}$, where $f$ is a computable function depending only on $k$~\cite{downey2012parameterized}. Problems admitting such algorithms are called \emph{fixed-parameter tractable} (\FPT)~\cite{cygan2015parameterized,flum2006parameterized}. 
For more details, we refer the reader to~\cite{cygan2015parameterized}.



\section{Polynomial Intractability} 
In this section, we study the classical computational complexity of the allocation problems under connectivity constraints. We establish hardness results for both egalitarian and utilitarian welfare objectives across several graph classes. Our results begin with hardness on highly restricted structures such as paths and then extend to broader graph families including trees, forests, and cycles. We also identify graph classes where efficient algorithms exist, thereby delineating the boundary between tractable and intractable cases.

We observe that the problem of finding a proportional connected allocation (PROP-CFD) is a strictly restricted special case of our egalitarian problem. In PROP-CFD, agent valuations are normalized such that \( \sum_{v \in V} u_i(v) = 1 \) for all \( i \in N \), and the target egalitarian threshold is fixed at \( t = 1/n \). 

Theorem 3.1 in \cite{bouveret2017fair} establishes that PROP-CFD is \NPC when $G$ is a path. While their hardness reduction is constructed using rational valuations, it extends to our integral definition: by scaling all valuations and the target threshold by the least common multiple of their denominators, we obtain a structurally identical instance where all values are integers. Because this uniform scaling preserves the validity of any allocation, our integral formulation strictly generalizes this \NPC problem.

\begin{corollary}
The \CEGW{} problem is \NPH, even when the graph \( G \) is a path.
\end{corollary}

Previous work by \cite{kawase2024contiguous} (Theorem~12) established that the \CUGW{} problem is \NPH even on paths.

\begin{corollary}
Both the \CUGW{} and \CEGW{} problems are \NPH when the graph \( G \) is a tree, a linear forest (collection of paths), or a forest (collection of trees).
\end{corollary}

Hardness for cycles reduces from paths. Given an instance on a path \( P \) with agents \( N \), we construct a cycle \( C \) by adding a single dummy vertex \( v^* \) connected to the two endpoints of \( P \). We introduce one additional \emph{blocker agent} \( B \) who values only \(v^*\). Also, \(v^*\) is valued only by \(B\). We set its valuation for \(v^*\) as:
\begin{itemize}
    \item \textbf{Egalitarian:} Set \( u_B(v^*) = t \), where \( t \) is the target threshold of the original path instance. Define the new threshold \(t'=t\) for the cycle instance.
    \item \textbf{Utilitarian:} Set \( u_B(v^*) = M \), where \( M = 1 + \sum_{i \in N} \sum_{v \in V(P)} u_i(v) \). Define the new threshold \(t'=t+M\) for the cycle instance.
\end{itemize}
In both settings, any valid solution (achieving threshold \( t \) or maximizing total welfare) strictly requires allocating \( v^* \) exclusively to agent \( B \). Because bundles must be connected, agents in \( N \) cannot wrap around the cycle and are restricted to valid connected bundles within the original path \( P \). Thus, we have the following lemma:

\begin{lemma}
Both the \CEGW{} and \CUGW{} problems are \NPH when the graph \(G\) is a cycle or a collection of disjoint cycles.
\end{lemma}

For star graphs, we can show that both problems are polynomial time solvable using the algorithm described in Theorem~3.2 of \cite{bouveret2017fair}. It gives a perfect bipartite matching algorithm to minimize cost to the agent who is allotted the star center. 

\begin{corollary}
Both the \CEGW{} and \CUGW{} problems are polynomial-time solvable when the  graph \(G\) is a star.
\end{corollary}

However, for collection of stars, we show in Section~\ref{sec:para_no_comp} this approach extends only for the utilitarian setting. For egalitarian welfare, it is \NPH for just 2 disjoint stars.
\section{Parameterization by Number of Agents }
\label{sec:no_of_agents}

In this section, we study both \CEGW{} and \CUGW{} parameterized by the number of agents \(n\). We design \FPT on several restricted graph classes and then extend these results to broader graph families. 

A recurring challenge is that many graph classes of interest naturally appear as disconnected collections of simpler structures. Rather than designing separate algorithms for each such graph family, we first establish a generic lifting result. The following meta-theorem shows that if \CEGW{} and \CUGW{} admit an \FPT algorithm on a single graph from a family \(\mathscr{F}\), then the same algorithms can be extended to disjoint unions of graphs from \(\mathscr{F}\) with only an additional parameter-dependent overhead.

\begin{theorem}\label{thm:meta-theorem}
Let $\mathscr{F}$ be the family of graphs where both the \CEGW{} and \CUGW{} can be solved in time $\OO(f(n)\cdot|V|^{\OO(1)})$, then both these problems can be solved on disjoint collection of graphs from $\mathscr{F}$ with an additional multiplicative overhead of $\mathcal{O}^{\star}(2^{\OO(n\log n)})$.
\end{theorem} 

\begin{proof}
Let $G=(V,E)$ be a collection of disjoint graphs from $\mathscr{F}$, with connected components 
$G_1,\dots,G_k $. We consider $\mathcal{F}(G)=\mathcal{F}_{\mathrm{conn}}(G)$. We show below how to extend an algorithm for a single member of $\mathscr{F}$ to a collection of graphs from $\mathscr{F}$, with an additional $\FPT(n)$ overhead.

Note that, since $G$ is a disjoint collection $G_1,\dots,G_k $, every feasible allocation 
$\mathcal{A}=(A_1,\dots,A_n)$ has the property that each bundle $A_i$ is contained entirely inside a single component $G_i$. Otherwise, $G[A_i]$ would be disconnected, contradicting feasibility. Thus, any feasible allocation induces a partition of the agent set $N=\{1,\dots,n\}$ into groups, where each group corresponds to agents served by the same component.

\noindent
\textbf{Enumerating Agent Partitions.}
We enumerate all ordered partitions of the agent set $\mathcal{P}=(S_1,\dots,S_x)$,
where the sets $S_1,\dots,S_x$ are nonempty and pairwise disjoint, and 
$1 \le x \le n$. The number of unordered partitions is the $n$-th Bell number, which is at most $2^{\OO(n\log n)}$. Applying ordering just adds a \(x! \le n!\) factor, and the total number of ordered partitions remains $2^{\OO(n\log n)}$. Since $n$ is the parameter, this enumeration depends only on $n$. Fix one such partition $\mathcal{P}$.
Next, we apply the {\em color-coding} technique~\cite{AlonYusterZwick1995} that helps us identify the best possible allocations respecting such a partition.  

\noindent
\textbf{Color coding.}
For a feasible allocation consistent with $\mathcal{P}$, we must choose $x$ distinct graph components
$G_{i}^1,\dots,G_{i}^{x}$ such that block $S_j$ is allocated entirely inside $G_{i}^{j}$. Instead of trying all $k^x$ assignments, we apply color-coding.
We randomly color each graph $G_\ell$ independently with one of the $x$ colors.
We call a coloring \emph{good} if there exists an assignment (in an optimal solution respecting the partitioning constraint) of $x$ 
distinct graphs $G_{i}^{1},\dots,G_{i}^{x}$ to the blocks $S_1,\dots,S_x$
such that $G_{i}^{j}$ receives color $j$ for every $j \in \{1,\dots,x\}$.
If such an assignment exists, then a random coloring is good with probability
at least $\frac{1}{x^x}$. In the following, we condition on a good coloring.

\noindent
\textbf{Solving colored instances.}
For each block $S_j$ and for each graph $G$ of color $j$, we compute the optimal welfare achievable inside $G$ when allocating only among agents in $S_j$. 
For this purpose, we invoke the $\OO(f(n)\cdot |V|^{\OO(1)})$ algorithm designed for a single component $G_i$.

\noindent
\textbf{Welfare.}
For each block $S_j$, we select the graph of color $j$ that maximizes the utilitarian/egalitarian welfare.

\noindent
\textbf{Running Time.}
We enumerate all partitions of the agent set, which contributes a factor $2^{\OO(n\log n)}$. 
For each partition with $x$ blocks, we repeat the color-coding procedure $x^x \le n^n$ times, contributing a factor $2^{\OO(n\log n)}$. Since for each component we need to spend $\OO(f(n)\cdot |V|^{\OO(1)})$, total time taken to solve the problem in $G$ is $\OO(2^{\OO(n\log n)} \cdot f(n)\cdot |V|^{\OO(1)})$

\noindent\textbf{Derandomization.} The randomization used in the color-coding step can be removed by employing a family of $x$-perfect hash functions (universal hash families) as in~\cite{AlonYusterZwick1995}. 
Such a family of size $2^{\OO(x)} k^{\log k}$ can be constructed in time $2^{\OO(x)} x^{\log x} \mathrm{poly}(k)$, and since $x \le n$, this preserves the running time $2^{\OO(n\log n)} \cdot \mathrm{poly}(|V|)$. Hence, the algorithm can be derandomized without affecting fixed-parameter tractability.
\qed \end{proof}

\subsection{FPT Algorithms for Trees and Forests}
Here we show that the \CEGW{} and \CUGW{} both are \FPT parameterized by the number of the agents using dynamic programming based approaches. Finally, using Theorem~
\ref{thm:meta-theorem}, we show that both these problems are \FPT for the same parameter. This also gives that these two problems are \FPT parameterized by the number of agents on paths and disjoint union of paths. 

\begin{theorem}\label{thm:egal-tree-fpt-n}
The \CEGW{} problem is \FPT parameterized by the number of agents~\(n\), when the graph is a tree.
\end{theorem}

\begin{proof}
The construction follows the dynamic programming framework of Theorem~3.4 in~\cite{bouveret2017fair},  
generalized to handle a target welfare threshold \( t \) instead of a fixed proportionality constant \( \frac{1}{n} \).

Root \(G\) at an arbitrary vertex \(r\).  
For \(v\in V\), let \(D(v)\) denote the subtree rooted at \(v\) and \(C(v)\) its children.  
Let \(N=\{1,\dots,n\}\) be the set of agents. 
For every vertex \(v\), agent \(i\), and set \(S\subseteq N\setminus\{i\}\) define
\[
\begin{aligned}
DP[v,i,S] &= \max\Bigl\{\, u_i(A_i) : 
  \mathcal{A} \text{ is a connected allocation of } D(v),\\
&\hspace{2.7cm}\text{the component of }\mathcal{A}\text{ containing }v\text{ is assigned to }i,\\
&\hspace{2.7cm}\text{and } u_j(A_j)\ge t\ \forall j\in S \Bigr\},
\end{aligned}
\]
and set \(DP[v,i,S] = -\infty\) if no such allocation exists.
If \(S=\emptyset\), then \(DP[v,i,\emptyset] = u_i(D(v))\). The goal is to check if \(DP[r, i, N \setminus i] \geq t\) for any agent \(i\). We take a bottom-up approach from leaf to root.

Let \(v\) have the set of children \(C(v) = \{z_1, \dots, z_d\}\).  
We enumerate all partitions \(\mathcal{P}\) of \(S\) into \(d\) (possibly empty) parts; say \(\mathcal{P} = \{P_1, \dots, P_d\}\). We now check if such a partition of agents to the subtrees rooted at the children is feasible.
For each partition \(\mathcal{P}\), construct the bipartite graph \(H_{\mathcal{P}} = (Z, Z', L)\) as follows:
\(
Z = \mathcal{P}, Z' = C(v).
\)
For every non-empty \(P \in \mathcal{P}\) and child \(z \in C(v)\), include edge \(\{P, z\}\in L\) iff the agents in \(P\) can be feasibly allocated within \(D(z)\): either agent \(i\) takes child \(z\), i.e.
\(
DP[z,i,P] \neq -\infty
\)
or 
\(\exists\, j\in P \text{ with } DP[z,j,P\setminus\{j\}] \ge t.
\)
When \(\{P, z\}\in L\), we set its weight
\(w(P,z) = DP[z,i,P]\)
if 
\(DP[z,i,P]\neq -\infty\) and 0 otherwise. The edge weight symbolizes contribution to the DP value we are computing at vertex \(v\).
For every empty \(P\in \mathcal{P}\) and every child \(z\in C(v)\), include edge \(\{P, z\}\) with weight
\(w(P,z) = u_i(D(z)).\) This represents agent \(i\) having to occupy the entire child's subtree. 

Note that for all edges in the bipartite graph \(H_{\mathcal{P}}\), their weights represent contribution to agent \(i\)'s utility.
If \(H_{\mathcal{P}}\) admits no perfect matching, define \(W(\mathcal{P}) = -\infty\).  
Otherwise, let \(W(\mathcal{P})\) be the maximum weight of a perfect matching in \(H_{\mathcal{P}}\). 
Finally, we have
\(
DP[v,i,S] = \max_{\mathcal{P}} \bigl( u_i(v) + W(\mathcal{P}) \bigr),
\)
where the maximum ranges over all partitions \(\mathcal{P}\) of \(S\) into at most \(|C(v)|\) parts.

The instance admits an allocation with minimum utility at least \(t\) iff \(\exists i\in N\) with
\(
DP[r,i,N\setminus\{i\}] \ge t.
\)

There are \(\OO(m \cdot n \cdot 2^n)\) DP states.  
For each state, the number of partitions of \(S\) depends only on \(n\), bipartite graph construction takes polynomial time in \(|\mathcal{P}|\), and computing a maximum-weight perfect matching for each partition takes polynomial time in \(|C(v)|\).  
Hence, the overall running time is \(f(n)\cdot (n+\mathrm{poly}(m))\) for some exponential function \(f\), proving fixed-parameter tractability with respect to \(n\).
\qed \end{proof}

\begin{theorem}\label{thm:util-tree-fpt-n}
The \CUGW{} problem is \FPT parameterized by the number of agents~\(n\), when the graph is a tree.
\end{theorem}

\begin{proof}
The core idea follows a similar dynamic programming approach as in Theorem~3.4 of \cite{bouveret2017fair}, suitably adapted for the utilitarian welfare objective.

We root the tree at an arbitrary vertex \(r\).  
For each node \(v\), agent \(i\), and subset of agents \(S \subseteq N \setminus \{i\}\), define
\(DP[v, i, S]\)
as the maximum total utilitarian welfare obtainable from the subtree \(D(v)\) rooted at \(v\), assuming that agent \(i\) receives node \(v\) (and hence the connected subgraph containing it) and that all agents in \(S\) are also allocated within this subtree. If \(S=\emptyset\), then \(DP[v,i,\emptyset] = u_i(D(v))\). If \(v\) is a leaf and \(|S|\geq1\), then \(DP[v, i, S] = -\infty\). We take a bottom-up approach from leaf to root.

Let \(v\) have the set of children \(C(v) = \{c_1, c_2, \dots, c_d\}\).  
We must allocate the remaining agents \(S \setminus \{i\}\) among these subtrees.
Consider a partition \(\mathcal{P} = \{P_1, P_2, \dots, P_d\}\) of \(S \setminus \{i\}\), where each subset \(P_j\) is (possibly empty) and corresponds to the set of agents assigned to child \(c_j\).

We define a bipartite graph \(H_{\mathcal{P}} = (Z, Z', L)\) with
\(Z = \mathcal{P}, Z' = C(v),\).
For any \(P \in \mathcal{P}\) and \(z \in C(v)\) the edge \((P, z)\) is included in \(L\) having weight
\[
w(P, z) =
\begin{cases}
\displaystyle
\max \Bigl\{
DP[z, i, P], \;
\max_{j \in P} DP[z, j, P \setminus \{j\}]
\Bigr\}, & \text{if } P \neq \emptyset, \\[1.2em]
u_i(D(z)), & \text{if } P = \emptyset.
\end{cases}
\]
Intuitively, for \(P \neq \emptyset\), the first term---\(DP[z, i, P]\) corresponds to the case where agent \(i\) continues as the root in the child subtree, and the second term---\(\max_{j \in P} DP[z, j, P \setminus \{j\}]\) corresponds to delegating that role to some agent \(j \in P\).

We compute a maximum-weight perfect matching in \(H_{\mathcal{P}}\),
\(
DP[v,i,S] = \max_{\mathcal{P}} \bigl( u_i(v) + W(\mathcal{P}) \bigr)
\).
The optimal utilitarian welfare is
\(
\max_{i \in N} DP[r, i, [n]]
\), 
and the complexity analysis is similar to the egalitarian case.
\qed \end{proof}

\noindent\textbf{Stars, Paths, Cycles.}
Since paths and stars are trees, we obtain the same \FPT algorithms parameterized by the number of agents as a direct consequence of Theorems~\ref{thm:egal-tree-fpt-n} and \ref{thm:util-tree-fpt-n}. For cycles, we can enumerate the exact allocation given to a single agent. Due to the connectivity requirement, this allocation must be either a contiguous subpath or the entire cycle, giving at most $\OO(|V|^2)$ choices. Combined with $n$ choices for which agent receives the allocation, the problem reduces to solving $\OO(|V|^2)\times n$ many instances on paths (where the remaining vertices form a path after removing the allocated connected subgraph with one fewer agent) immediately implying the following corollary.

\begin{corollary}\label{cor:path-fpt-n}
Both the \CEGW{} and \CUGW{} problems are \FPT parameterized by the number of agents \(n\) when the graph is a path, star or a cycle.
\end{corollary}

Thus Theorem~\ref{thm:meta-theorem} combined with Theorem \ref{thm:egal-tree-fpt-n}, \ref{thm:util-tree-fpt-n} and Corollary \ref{cor:path-fpt-n}, results in the following stronger result.

\begin{theorem}
Both the \CEGW{} and \CUGW{} problems are \FPT parameterized by the number of agents, when the graph is a forest or disjoint collection of cycles.
\end{theorem}

\subsection{On Graphs with Bounded Treewidth and Treedepth}

We now study the tractability of the \CEGW{} and \CUGW{} problems under the combined parameterization by the number of agents \(n\) and structural graph parameters such as treewidth and treedepth. Here, the two welfare objectives exhibit contrasting behavior. We show that the egalitarian variant remains computationally hard even on highly restricted graph families of bounded treedepth, implying hardness even when combined with the number of agents parameter. In contrast, the utilitarian variant becomes \FPT when parameterized by the number of agents together with treewidth, and consequently also with treedepth.

We first show hardness for the egalitarian objective, which already holds on complete bipartite graphs $K_{2,q}$ of treedepth two.

\begin{theorem}[$\clubsuit$]\footnote{ Proofs of statements marked with ($\clubsuit$) move to Appendix}\label{thm:egal-hard-td}
   The \CEGW{} problem is \NPH even when the graph has treedepth two.
\end{theorem}

Hence, the \CEGW{} problem becomes {\sf para}-\NPH parameterized by the number of agents. Moreover, since the treedepth of a graph is an upper bound on the treewidth of the graph, the same intractability result in the following corollary.

\begin{corollary}
The \CEGW{} problem is {\sf para}-\NPH parameterized by the number of agents, even when the graph has a bounded treewidth.
\end{corollary}

Next we focus on the utilitarian version of the problem and show that the problem becomes \FPT parameterized by the number of agents, when the graph has a bounded treewidth. 

\begin{theorem}[$\clubsuit$]\label{thm:util-fpt-tw}
   The \CUGW{} problem is \FPT parameterized by the number of agents, when the graphs has a bounded treewidth.
\end{theorem}

As discussed earlier in this section, the treedepth of a graph is an upper bound on the treewidth of the graph; the \FPT result holds true for the treedepth parameter as well, which results in the following corollary. 

\begin{corollary}\label{cor:util-fpt-td}
   The \CUGW{} problem is \FPT parameterized by the number of agents, when the graph has a bounded treewidth.
\end{corollary}


\section{Parameterization by Number of Components}
\label{sec:para_no_comp}
Note that the \NPH{ness} of both the \CEGW{} and \CUGW{} problems on paths, cycles, and trees immediately implies their {\sf para NP}-hardness (impossibility of even \XP algorithms) when parameterized by the number of connected components, even when every connected component belongs to one of these graph classes. We summarize this implication below.
\begin{corollary}
Both the \CEGW{} and \CUGW{} problems are {\sf para NP}-hard parameterized by the number of connected components, even when every connected component is either a path, a cycle, or a tree.
\end{corollary}
Thus, the only remaining graph class to consider is stars, for which the above implication does not immediately apply, since the problems are polynomial time solvable on a star.

In this section, we establish hardness results for this remaining case. We first show that the \CEGW{} problem remains \NPH even when the graph consists of only two disjoint stars. This immediately yields {\sf para NP}-hardness with respect to the number of stars, and consequently rules out the existence of an \XP algorithm under standard complexity assumptions.
For the utilitarian variant, the \CUGW{} problem, we obtain a slightly weaker hardness result: the problem is \WTH when parameterized by the number of stars. In contrast to the egalitarian version, however, we also design an \XP algorithm for this parameterization. We start with the egalitarian version on stars below.

\begin{theorem}\label{thm:egal-hard-stars}
The \CEGW{} problem is \NPH when the graph \( G \) consists of two disjoint stars.
\end{theorem}
\begin{proof}
We reduce from the \textsc{Partition} problem. Given a multiset of integers \( S = \{x_1, \dots, x_q\} \) summing to \( 2b \), we ask if there exists a subset summing to \( b \).

 We construct \( Star_1 \) with center \( c_1 \) and leaves \( y_1, \dots, y_q \), and \( Star_2 \) with center \( c_2 \) and leaves \( z_1, \dots, z_q \).
We create the following \( q+2 \) agents:
2 \emph{star agents} \( I_1, I_2 \), and \( q \) \emph{item agents} \( J_1, \dots, J_q \).
\( I_1 \) values \( c_1 \) at \( b+1 \) and \( y_i \) at \( x_i\) \(\forall i \in [q]\).
\( I_2 \) values \( c_2 \) at \( b+1 \) and \( z_i \) at \( x_i\) \(\forall i \in [q]\).
 Agent \( J_i \) values \( y_i \) and \( z_i \) at \( t \) \(\forall i \in [q]\), where \( t = 2b+1 \).
Set the global threshold to \( t \). In any valid allocation achieving threshold \( t \), 
Star Agents must take the centers (to get utility \(\ge 2b+1\)). Otherwise the maximum utility by taking a leaf is \(2b\).
Each Item Agent \( J_i \) must take exactly one leaf (either \( y_i \) or \( z_i \)) to reach utility \( t \).
This creates a partition of indices \( I \subseteq \{1, \dots, q\} \) where \( J_i \) takes \( y_i \) \(\forall i \in I\) (from \( Star_1 \)).
\( I_1 \) keeps the remaining leaves of \( Star_1 \). To reach utility \( t = 2b+1 \), the sum of these remaining leaves must be at least \( b+1 \). Similarly for \( I_2 \).
Since the total leaf weight is \( 2b \), both Star Agents are satisfied if and only if the remaining leaves sum exactly to \( b \), which implies a valid solution to \textsc{Partition}. The other direction is direct.
\qed \end{proof}

We now turn to the utilitarian variant of the problem. We show that the \CUGW{} problem admits an \XP algorithm when parameterized by the number of stars.

\begin{theorem}
The \CUGW{} problem is in \XP with respect to the number of disjoint stars \( k \).
\end{theorem}

\begin{proof}
Let \( G \) be the disjoint union of \( k \) stars \( S_1, \dots, S_k \), where star \( S_j \) has center \( C_j \) and a set of leaves \( L(S_j) \). Let \( N \) be the set of \( n \) agents.
We propose an algorithm that proceeds as follows. First, we generate all possible ordered sequences of \( k \) distinct agents from \( N \) to serve as Center Owners for \( S_1, \dots, S_k \). There are \( P(n, k) = \frac{n!}{(n-k)!} \) such permutations. For each such assignment of center owners \( \sigma: \{1, \dots, k\} \to N \), we check if a feasible allocation can be obtained.
Next, we define the value \(BUW_\sigma\) as the total utilitarian welfare obtained if the \(k\) center owners also occupy all the adjacent leaves:
\[
BUW_\sigma = \sum_{j=1}^k u_{\sigma(j)}(C_j) + \sum_{j=1}^k \sum_{v \in L(S_j)} u_{\sigma(j)}(v).
\]
We now check if welfare improves by reassigning leaves to the remaining agents \( N' = N \setminus \{\sigma(1), \dots, \sigma(k)\} \). Construct a bipartite graph \(H_\sigma = (Z, Z', L)\) with \( Z = N' \) and \( Z' = \bigcup_{j=1}^k L(S_j) \). For each agent \( i \in Z \) and leaf \( v \in Z' \) (where \( v \in L(S_j) \)), add an edge with weight \(w_{iv} = u_i(v) - u_{\sigma(j)}(v)\). This is the change in \(BUW_\sigma\) that occurs if the leaf is given to agent \(i\) instead of the center owner \(\sigma(j)\). Compute the Maximum Weight Matching \( M \) in this bipartite graph (considering only edges with \( w_{iv} > 0 \)). Because we only consider edges with strictly positive weights, any valid matching $M$ corresponds to a welfare-improving reassignment where matched leaves are transferred to matched agents in $Z$, and unmatched leaves implicitly remain with their respective center owners. The maximum welfare for the center assignment \(\sigma\) is \( BUW_\sigma + w(M) \).
Finally, we return the maximum welfare found over all permutations \(\sigma\). The number of permutations is \( \OO(n^k) \). For each permutation, we solve a Maximum Weight Bipartite Matching problem, which takes polynomial time in the number of vertices and edges. The total time complexity is \( \OO(n^k \cdot \mathrm{poly}(n, m)) \).
\qed \end{proof}

A similar approach fails for egalitarian welfare, where one must simultaneously guarantee a minimum payoff for every center owner.

\begin{theorem}[$\clubsuit$]\label{thm:util-whard-stars}
The \CUGW{} problem is \WTH with respect to the number of disjoint stars $k$ even when the input graph is restricted to a disjoint union of stars.
\end{theorem}

\bibliographystyle{splncs04}
\bibliography{references}
\newpage
\appendix
\setcounter{page}{1}
\section*{Appendix}
This is appendix for submission ID: 58. We provide the missing proofs of some of the theorems.

\section{Proof of Theorem~\ref{thm:egal-hard-td}}
\begin{proof}
We reduce from the classical \textsc{Partition} problem. An instance of \textsc{Partition} consists of positive integers 
$a_1,\dots,a_q$ with total sum $2B$, and asks whether there exists a subset $I \subseteq [q]$ such that
$\sum_{i \in I} a_i = B.$ Given such an instance, we construct an equivalent instance of egalitarian welfare maximization as follows.

Let $G$ be the complete bipartite graph $K_{2,q}$.
Let the two vertices on the left side be $u_1$ and $u_2$, and let the $q$ vertices on the right side be $v_1,\dots,v_q$. Observe that $K_{2,q}$ has treedepth $2$.
There are two agents, $1$ and $2$.
For each $i \in [q]$: both agents assign value $a_i$ to vertex $v_i$ and $0$ to $u_1$ and $u_2$.

The given \textsc{Partition} instance is a yes-instance if and only if
the maximum egalitarian welfare in the constructed instance is at least $B$.

Since $G$ is $K_{2,q}$, any connected bundle allocated to an agent must consist of
one of the vertices $u_1$ or $u_2$ together with an arbitrary subset of the
vertices $v_i$ adjacent to it. Thus, any feasible allocation corresponds to a
partition of $\{v_1,\dots,v_q\}$ into two sets, one assigned to $u_1$ and the other
to $u_2$.

Under any connected allocation, the value received by agent $1$ equals
$\sum_{i \in I} a_i$ for some subset $I \subseteq [q]$, and agent $2$
receives $\sum_{i \notin I} a_i$.
Since $\sum_{i=1}^q a_i = 2B$, the egalitarian welfare is
$\min\left\{\sum_{i \in I} a_i,\; 2B - \sum_{i \in I} a_i\right\}$.
This quantity is at least $B$ if and only if
$\sum_{i \in I} a_i = B$.
Hence, there exists a connected allocation with egalitarian welfare at least $B$ if and only if the original \textsc{Partition} instance is a yes-instance.
\qed\end{proof}

\section{Proof of Theorem~\ref{thm:util-fpt-tw}}

\begin{proof}
Let $V = V(G)$ and $E = E(G)$. Let $(\tau, B)$ be a nice tree decomposition of $G$ of width at most $\tw$, rooted at $r$. For a node $t \in V(\tau)$, let its bag be $B_t$, where $|B_t| \leq \tw + 1$. For a node $t \in V(\tau)$, denote by $V_t$ the set of vertices in the bags of the subtree of $\tau$ rooted at $t$, and let $G_t = G[V_t]$.
We use a bottom-up dynamic programming approach as follows.




\medskip


\noindent\textbf{DP States.}
We maintain a DP table indexed by states $\text{DP}[t, \sigma, \pi]$, where:
\begin{itemize}
    \item $t \in V(\tau)$ is a node of the tree decomposition.
    \item $\sigma : B_t \to N$ is an assignment function mapping bag vertices to agents $N = \{1, \ldots, n\}$.
   \item $\pi = (\pi_1, \ldots, \pi_n)$ is a connectivity profile. For each agent $i$, $\pi_i$ is a partition of $B_t^i := \{v \in B_t \mid \sigma(v) = i\}$ that describes which boundary vertices belong to the same connected component of the partial assignment allocated to agent $i$ so far.
\end{itemize}

\noindent\textbf{Interpretation of States.}
The value $\text{DP}[t, \sigma, \pi]$ represents the maximum utilitarian welfare achievable by a partition $(A_1^t, \ldots, A_n^t)$ of $V_t$ such that:
\begin{enumerate}
    \item $A_1^t, \ldots, A_n^t$ partition $V_t$.
    \item For every agent $i$, we have $A_i^t \cap B_t = \{v \in B_t \mid \sigma(v) = i\}$.
    \item For every agent $i$, two vertices $u, v \in B_t^i$ are in the same block of $\pi_i$ if and only if they are in the same connected component of $G_t[A_i^t]$.
\end{enumerate}

If no such partition exists, $\text{DP}[t, \sigma, \pi] = -\infty$.









\noindent\textbf{DP Recurrence.}
We describe the recurrence for each of the four node types in the nice tree decomposition.

\paragraph{\textbf{Leaf node.}}
In a nice tree decomposition, the leaf bag is empty: $B_t=\emptyset$. There is exactly one state and we set
$
\mathrm{DP}[t,\emptyset,\emptyset]=0.
$

\paragraph{\textbf{Introduce vertex node.}}
Suppose $t$ introduces a vertex $v$ and has a single child $t'$ with
$
B_t = B_{t'} \cup \{v\}.
$
For a state $(t,\sigma,\pi)$ and let $\sigma'$ be the restriction of
$\sigma$ to $B_{t'}$. Let $i = \sigma(v)$ be the agent to whom $v$ is assigned. We consider child states $(t',\sigma',\pi')$ satisfying:
\begin{itemize}
    \item For every agent $j \neq i$, we have $\pi_j = \pi'_j$.
    \item The partition $\pi_i$ is obtained from $\pi_i'$ by: (1) adding $\{v\}$ as a singleton block, and (2) for every vertex $u \in B_{t'}$ with $\sigma(u) = i$ and $(u, v) \in E$, merging the blocks containing $u$ and $v$.
\end{itemize}
If these conditions hold, the transition is feasible and we set:
\[
\text{DP}[t, \sigma, \pi] = \max_{\sigma', \pi'} \left( \text{DP}[t', \sigma', \pi'] + u_i(v) \right),
\]
where the maximum is over all compatible child states. Otherwise, the state is infeasible and we set the respective DP entry to -$\infty$.

\paragraph{\textbf{Forget vertex node.}}
Suppose $t$ forgets vertex $v$ and has child $t'$ with
$B_t = B_{t'} \setminus \{v\}$. 
At a forget node, the vertex $v$ disappears from the bag and will never
reappear in any ancestor bag. Therefore, this is the last opportunity to
verify that the connected component of agent $i=\sigma'(v)$ containing $v$
can still be connected through the remaining boundary vertices.

Let $C$ be the block of $\pi'_i$ that contains $v$. If $C=\{v\}$ while $\pi'_i$ contains other blocks, then the component of $v$ becomes permanently separated from the rest of agent $i$'s bundle, and connectivity in the final solution would be impossible. Hence such states are discarded.

Otherwise, either $v$ is still connected to other boundary vertices of agent
$i$ (i.e., $C\setminus\{v\}\neq\emptyset$) or it was the only component of
agent $i$ in the subtree (i.e., $|\pi'_i|=1$), and the state remains feasible.
Fix a state $(t,\sigma,\pi)$. We consider child states
$(t',\sigma',\pi')$ satisfying the following.

\begin{itemize}
    \item The assignment $\sigma'$ extends $\sigma$, that is,
    $\sigma'|_{B_t}=\sigma$.
    \item For every agent $j\neq i$ (where $i=\sigma'(v)$),
    we have $\pi_j=\pi'_j$ after removing $v$.
    \item The partition $\pi_i$ is obtained from $\pi'_i$ by
    deleting $v$ from the partition.
\end{itemize}

Let $i=\sigma'(v)$ and let $C$ be the block of $\pi'_i$
that contains $v$.
The transition is allowed only if $C \setminus \{v\} \neq \emptyset$ or $|\pi_i'| = 1$. If feasible, we set:
\[
\text{DP}[t, \sigma, \pi] = \max_{\sigma', \pi'} \left( \text{DP}[t', \sigma', \pi'] \right),
\]
where the maximum is over all compatible child states. Otherwise, the state is infeasible and we set the respective DP entry to -$\infty$.

\paragraph{\textbf{Join node.}}
Suppose $t$ is a join node with children $t_1$ and $t_2$ such that
$
B_t = B_{t_1} = B_{t_2}.
$
At a join node, we combine two disjoint subtrees sharing the common boundary $B_t$. For each agent $i$, both children specify which boundary vertices are connected within their respective subtrees. When combining, two boundary vertices of agent $i$ are considered connected in the combined tree if they are connected in at least one child.

For each agent $i$, we merge the connectivity information from both children as follows: starting with blocks of $\pi_i^{(1)}$ and $\pi_i^{(2)}$, we repeatedly merge any two blocks sharing a common vertex until no further merges are possible, yielding the partition $\widehat{\pi}_i$.
This resulting partition describes the connectivity in the combined
subtree.

For a state $(t, \sigma, \pi)$, we consider pairs of child states $(t_1, \sigma, \pi^{(1)})$ and $(t_2, \sigma, \pi^{(2)})$. The pair is compatible only if $\widehat{\pi}_i = \pi_i$ for every agent $i \in N$. If compatible, we set:
\[
\text{DP}[t, \sigma, \pi] = \max \left( \text{DP}[t_1, \sigma, \pi^{(1)}] + \text{DP}[t_2, \sigma, \pi^{(2)}] - \sum_{v \in B_t} u_{\sigma(v)}(v) \right),
\]
where the maximum is over all compatible pairs. Since vertices of the bag are counted in both children, we subtract
their utilities once to avoid double counting.

\paragraph{\textbf{Answer.}}
The optimal utilitarian welfare is stored against $\text{DP}[r, \emptyset, \emptyset]$.

\paragraph{\textbf{Complexity.}}
The number of DP states is $O(|\tau| \cdot n^{\tw+1} \cdot S)$, where $S$ is the number of possible connectivity profiles per state and $|S|$ is  $O(\tw^\tw)$. For each state, computing the DP value requires polynomial time (bipartite matching at join nodes). 
Thus, the algorithm runs in time
\(
\OO\!\left(
|V|\cdot n^{\tw+1}\cdot 2^{\OO(\tw\log \tw)}
\right),
\)
proving fixed-parameter tractability with respect to the combined parameter \((\tw,n)\).
\qed\end{proof}

\section{Proof of Theorem~\ref{thm:util-whard-stars}}

\begin{proof}
We reduce from the \textsc{Parameterized Set Cover} problem, which is \WTC.
It says, given a universe \( U = \{1, \dots, q\} \), a family of sets \( \mathcal{S} = \{S_1, \dots, S_r\} \), and an integer \( k \), does there exist a subfamily \( \mathcal{S}' \subseteq \mathcal{S} \) of size \( k \) covering \( U \)?

We construct from the above, an instance of \CUGW{} with \( k \) stars. Construct \( k \) disjoint stars. Each star \( j \in [k] \) consists of a \emph{center} \( c_j \), \( q \) \emph{element leaves} \( \{l_{j,1}, \dots, l_{j,q}\} \) corresponding to the universe \( U \). For each set \( S_i \in \mathcal{S} \), we add a set of \( |S_i| \) \emph{padding leaves} \( \{d_{j,i,1}, \dots, d_{j,i,|S_i|}\} \). Hence each star \(j\) has \(q+\sum_{z=1}^{r}|S_z|\) leaves.
We define two types of agents with valuations in \( \{0, 1\} \). 
First, there are \( q \) \emph{element agents} \( \{B_1, \dots, B_q\}\) representing each element in \(U\). They value their element leaf in each star as 1, and all other vertices as 0. Formally, for agent \(B_x\), \(\forall j \in [k]\), \( u_{B_x}(l_{j,x}) = 1 \). 
Second, there are \( r \) \emph{set agents} \( \{C_1, \dots, C_r\} \) representing the sets in \( \mathcal{S} \). They value all centers, padding leaves for their own set, and complement element leaves with 1, and all other vertices as 0. Formally, for agent \(C_i\), \(\forall j \in [k], z \in [|S_i|], \forall x \in [q], x \notin S_i\), we define \( u_{C_i}(c_j) = 1 \), \( u_{C_i}(d_{j,i,z}) = 1 \), and \( u_{C_i}(l_{j,x}) = 1 \).
We set the target welfare for the utilitarian problem as \( t= r + kq + q \). 
It remains to show that a total welfare of \( t \) is achievable if and only if there exists a Set Cover of size \( k \).

\( (\Rightarrow) \) \emph{Set Cover exists.}
Let \( \mathcal{I} \) be the indices of the \( k \) sets in the cover.
Assign the \( k \) agents \( \{C_i \mid i \in \mathcal{I}\} \) to the \( k \) centers. The component allocated to them can be the center with padding leaves and complement leaves for the set. Thus each achieves utility \( 1 + |S_i| + (q - |S_i|) = 1 + q \), and the total utility for k chosen sets is \( k(1+q) \).
Assign the remaining \( r-k \) Set Agents to one of their respective padding leaves. This can be done by choosing any one star and picking one leaf from the required sets of padding leaves. From this, a utility of \( r-k \) is obtained.
For the occupied stars, the element leaves \( \{l_{j,x} \mid x \in S_i, \forall j \in [k]\} \) remain free. Since the union covers \( U \), every \( x \) is covered by some selected set, so at least one leaf \( l_{j,x} \) is free for every \( x \). Assign each \( B_x \) to such a leaf. Total utility from element agents is \(q\). We have a net utility of \(t\) from this allocation.

\( (\Leftarrow) \) \emph{Valid utilitarian allocation exists.}
We analyze the maximum possible utility that each group of agents can contribute to the global welfare. A set agent can achieve a maximum utility of \( 1+q \) (by taking a center). If not assigned a center, they can take at most one leaf due to connectivity constraints, yielding utility at most 1. Since the graph has only \( k \) centers, at most \( k \) set agents can achieve the higher utility. Thus, the total utility of the \( r \) set agents is strictly bounded by:
\(u_{\text{C}}^{\max} = k(1+q) + r-k = r + kq.\)
There are \( q \) Element Agents which can only assume a single leaf due to their valuation scheme and connectivity constraints. Their total utility is strictly bounded by:
\(u_{\text{B}}^{\max} = q.\)
The target welfare is exactly the sum of these upper bounds: \( t = u_{\text{B}}^{\max} + u_{\text{C}}^{\max} \).
Therefore, any allocation achieving \( t \) must simultaneously maximize the utility of \emph{both} groups.

Maximizing total valuation for set agents requires exactly \( k \) of them to occupy the centers. Let these be the ``Selected'' sets.
Maximizing welfare for element agents requires every agent \( B_x \) to receive utility 1. This implies that for every \( x \in U \), there exists \(j\) for which the leaf \( l_{j,x} \) is left free by the center owner. Leaf \( l_{j,x} \) is free only if the center owner \( C_i \) has \( x \in S_i \).
Thus, the condition that all \( q \) element agents are satisfied implies that \( \bigcup_{i \in \text{Selected}} S_i = U \).
The selected agents correspond to a valid Set Cover of size \( k \).
\qed \end{proof}
\end{document}